\documentclass[11pt]{extarticle}
 
\usepackage{graphicx} %
\usepackage[inkscapelatex=false]{svg}

\usepackage[letterpaper,top=2cm,bottom=2cm,left=2cm,right=2cm,marginparwidth=3.5cm]{geometry}
\usepackage{fancyhdr}
\usepackage{datetime}

\usepackage{amsmath}
\usepackage{amssymb}
\usepackage{amsthm}
\usepackage{amsfonts}
\usepackage{graphicx}
\usepackage{subcaption} 
\usepackage{xspace}
\usepackage[colorlinks=true, allcolors=blue]{hyperref}
\usepackage{enumerate}
\usepackage{color}
\usepackage{url}
\usepackage{centernot}
\usepackage{xspace}
\usepackage{hyperref}
\usepackage{enumitem}

\newcommand{\cF}{\mathcal{F}}

\newcommand{\cT}{\mathcal{T}}

\newcommand{\cS}{\mathcal{S}}
\newcommand{\cP}{\mathcal{P}}
\newcommand{\cE}{\mathcal{E}}
\newcommand{\cR}{\mathcal{R}}

\newcommand{\Exp}{\mathbb{E}}
\newcommand\eps{\varepsilon}

\newcommand\R{\mathbb{R}}
\newcommand\defnemph[1]{\underline{\emph{#1}}}
\newtheorem{theorem}{Theorem}
\newtheorem{lemma}[theorem]{Lemma}
\newtheorem{proposition}[theorem]{Proposition}

\newtheorem*{defn}{Definition}
\newtheorem*{theorem*}{Theorem}
\newtheorem*{lemma*}{Lemma}

\newtheorem{question}{Question}

\theoremstyle{remark}

\newcounter{marginalnote}

\newcounter{marginalnoteb}

\title{Approximately Decoding the  Colour Code}
\author{Mark Walters}
\date{\today}

\begin{document}
\maketitle
\begin{abstract}
Recently we showed that minimum weight decoding in the (6.6.6~planar) colour code is
NP-hard. However, it remained an open question as to whether it was
possible to approximate the minimum weight decoding arbitrarily closely in polynomial time. In this
paper we prove that it is possible: for any $\eps>0$ there is an
polynomial time algorithm that, given a syndrome, can find an error-set
generating that syndrome whose weight is at most $1+\eps$ times the
weight of the minimum weight decoding. As a consequence we see that, for any
$\eps>0$, there is a polynomial time algorithm that can
correct all errors of weight up to $(1-\eps)d/2$ in the distance $d$
colour code (so almost up to the theoretical $d/2$ limit).

The polynomial we give is impractically large, but it does
open the door for sensible polynomial time algorithms that approximate
minimum weight decoding and, in particular, shows that approximate decoding is
not NP-hard.
\end{abstract}
\section{Introduction}
\begin{figure}[t]
   \begin{subfigure}[t]{0.5\columnwidth}
     \includegraphics[width=1\columnwidth]{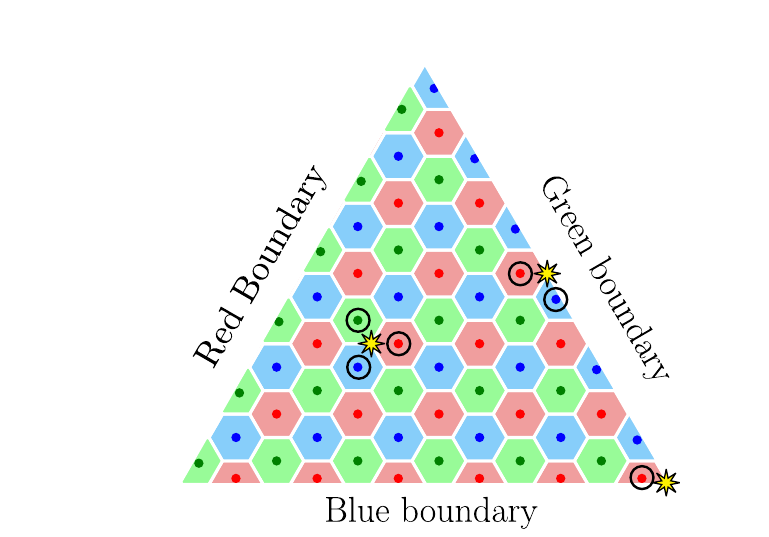}%
     \caption{\label{fig:bare-lattice-non-dual}}
   \end{subfigure}
   \begin{subfigure}[t]{0.5\columnwidth}
     \begin{center}
       \reflectbox{\includegraphics[width=0.55\columnwidth,angle=90]{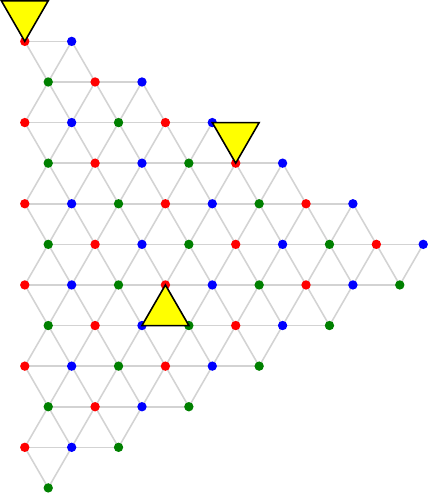}}%
     \caption{\label{fig:bare-lattice-dual}}
     \end{center}
   \end{subfigure}

  \caption{\label{fig:bare-lattice}\textbf{The colour code and its
      error model}. \textbf{(a)} The $6.6.6$ colour code (with
    boundary). There are data qubits on the vertices of a hexagonal
    lattice and the checks are the coloured faces.  An $X$-error on a
    bulk data qubit (the yellow star in the middle of the picture)
    triggers the three $Z$-checks shown ringed in black. However, an
    $X$-error on a qubit lying on a boundary, only triggers the two
    checks (faces) containing it (for example the yellow star on the
    middle right boundary); and if the error occurs it at a corner, it only triggers the single
    check that it is contained in (the yellow star in the bottom
    right). \textbf{(b)} The dual lattice of the colour code shown in
    \textbf{(a)}. In this figure the nodes correspond to checks, and
    the faces to data qubit errors. The three yellow
    triangles represent the data qubit errors corresponding to the
    errors shown in \textbf{(a)}.
  }
\end{figure}
The colour code is one of the most promising options for quantum error
correction with many advantages over several of the alternatives --
for example, logical operations are often easier to
implement. However, decoding has proved difficult -- in stark contrast
to the surface code. We discuss existing decoders briefly in
section~\ref{ss:other-decoders} -- for a more complete discussion of
both decoders and the advantages and disadvantages of the colour code
see our previous paper~\cite{waltersMinimumWeightDecoding2026}.

In that paper, we proved that decoding the hexagonal planar colour
code (see Figure~\ref{fig:bare-lattice} and Section~\ref{ss:cc})
is NP-hard -- that is, under the standard complexity assumption that
$\text{P}\not=\text{NP}$, there is no polynomial time algorithm for
finding the minimum weight error-set generating a given syndrome. (We
showed this is hard even in the simplest case of the code capacity
model with $X$-errors.)  However, one of the open questions asked in
that paper is whether it is possible to approximate the minimum weight
decoding in polynomial time, or whether such an approximation is
NP-hard. This is a natural question as many NP-hard problems are also
NP-hard to approximate. For example, it is NP-hard to approximate
max-3SAT which asks for the maximum number of clauses in a 3-SAT
formula that can be satisfied -- this is a case of the famous PCP
theorem~\cite{AroraSafra1998,arora1998proof} (with a tighter result
proved by
H\r{a}stad~\cite{hastadOptimalInapproximabilityResults2001}). In
contrast, for some NP-hard problems, it is possible to approximate
arbitrarily closely in polynomial time. One of the most notable
examples, and much of the inspiration for this paper, is Arora's
result~\cite{aroraPolynomialTimeApproximation1998} for the Euclidean
Travelling Salesman Problem.

In this paper we show that the latter case holds --
we give a polynomial
time algorithm approximating the minimum weight
decoding. Unfortunately, our polynomial is impractically large, but it
does show that approximating the minimum weight error-set is not
NP-hard. In other words, the
NP-completeness result
in~\cite{waltersMinimumWeightDecoding2026} does not rule out the
possibility of fairly accurate `reasonable' polynomial time
algorithms.

The precise statement of our result is given by the following theorem.

\begin{theorem}\label{t:cc-bounded}
  For any $\eps>0$ there is an algorithm $\mathbf{A_\eps}$
 such that given any syndrome $\cS$ in the distance $d$ 
  code-capacity colour code,
  \begin{itemize}
  \item
  $\mathbf{A_\eps}$ will output an error-set
  $\cE$ generating $\cS$, with $|\cE|\le (1+\eps)OPT$ where $OPT$ is
    the size of minimal error set generating $\cS$.
  \item The computational complexity (i.e. run time) of
    $\mathbf{A_\eps}$ is bounded by a polynomial in $d$.
  \end{itemize}
\end{theorem}
It is important to note that this is not the same as a heuristic or
average case decoder -- for \emph{every} syndrome, the algorithm given by
Theorem~\ref{t:cc-bounded}  approximates the minimum
weight decoding to within a factor of $1+\eps$.  We remark that the
algorithm $\mathbf{A_\eps}$ can correctly decide the logical effect of
all errors of up to weight $(1-\frac{\eps}{2})\frac{d}{2}$ -- i.e.,
nearly up to the theoretical threshold of $d/2$. This follows since,
if the sum of the weights of the true error set and the decoder's
solution error-set is less than the code distance $d$, then the
decoders solution and the true error-set must be logically equivalent
-- and, given any error-set of weight less than
$(1-\frac{\eps}{2})\frac{d}{2}$, the algorithm gives an error-set of
weight at most $(1+\eps)(1-\frac{\eps}{2})\frac{d}{2}<
(1+\frac{\eps}{2})\frac{d}{2}$ which together have sum less than $d$.

 Whilst the theorem is says that we can approximate the minimum weight
 decoding to within a multiplicative factor of $1+\eps$ in polynomial
 time, it is important to note that the algorithm and
 run-time depend on how closely we wish to approximate the minimum
 weight decoding.  Moreover, as we reduce $\eps$ the complexity of
 $\mathbf{A_\eps}$ the computational complexity increases rapidly --
 all of the degree, the leading coefficient and the constant term grow
 quickly -- for example, to approximate within $1\%$ our algorithm is
 polynomial but with degree over 100,000!  Whilst our bounds are far
 from tight -- we have focused on keeping the algorithm and proof
 simple (see Section~\ref{s:further-work} for discussion of possible
 improvements) -- it seems likely that substantial new ideas would be
 needed to make these algorithms practical.

In the bulk of this paper we restrict to showing that there is a polynomial time
algorithm for the simplest case of code-capacity with the only errors
being $X$-errors which all occur with the same probability. Our
results can be extended to errors from depolarising noise, unequal probabilities, circuit
level noise and other codes.  We discuss this in
Section~\ref{s:further-work}.

\textbf{Note.} We would like to bring the reader's attention to independent and
concurrent work by S.~Gu, L.~Wang and A.~Kubica~\cite{guPTAS2026}, which
provides a PTAS for minimum-weight decoding of topological codes with
point-like excitations connected by string-like errors.

\subsection{Comparison with previous work}\label{ss:other-decoders} 
There are many  decoders for the colour code -- these include  restriction
decoders~\cite{delfosseDecodingColorCodes2014,
  kubicaEfficientColorCode2023}, M\"obius
decoders~\cite{sahayDecoderTriangularColor2022,
  gidneyNewCircuitsOpen2023},  machine learning
decoders~\cite{seniorScalableRealtimeNeural2025}, and BP-based
decoders~\cite{koutsioumpasColourCodesReach2025}. Whilst many of these
work very well in practice, they typically do not give formal
guarantees on the approximation (note that a decoder can have an
excellent logical error rate whilst not meeting any approximation
bound -- good behaviour `on average' is sufficient). One exception, is
the restriction decoder of Delfosse and
Kubica~\cite{kubicaEfficientColorCode2023} which, as observed by Gu,
Wang and Kubica in~\cite{guColorCodeSurface2026}, does guarantee to
give a decoding which is at most three times the minimum
weight -- so it is possible to approximate within a
constant factor in `sensible' polynomial time.

\subsection{Terminology, Notation and the colour code.}\label{ss:cc}
In this section we briefly recall the (standard) terminology we use,
the notation we use, and briefly describe the (hexagonal) colour code.

The infinite colour code is an infinite hexagonal lattice, where each
vertex corresponds to a data qubit and each hexagonal face corresponds
to both an X-check and a Z-check (see
Figure~\ref{fig:bare-lattice-non-dual} for the finite colour code).
Since each data qubit is in three checks (of each type) an X-error
will trigger the three Z-checks containing it.  More generally, a
check is triggered if it meets an odd number of errors and we call
such a check a defect. The syndrome $\cS$ is the state of all the
checks. We remark that, whilst we have called this a code, the
infinite colour code does not actually encode any logical
qubits. However, the homogeneity (lack of boundary) makes it simpler
to work in, and still captures many of the key ideas, so we will start
by proving our result in this setting and then extend it to bounded
colour codes (see below).

We say that an error-set $\cE$ generates a
syndrome $\cS$ if the errors in $\cE$ would give $\cS$ as syndrome. A
minimum weight decoding of a syndrome $\cS$ is a set of errors $\cE$
that generates $\cS$, and has the minimum weight over all such sets.

We will work in the dual lattice (see Figure~\ref{fig:bare-lattice-dual}) where
each face (check) is replaced by a node, and each node (qubit) is
replaced by a (triangular) face.

The colour code with boundary is shown in
Figure~\ref{fig:bare-lattice} -- the difference is that we restrict to a
finite set of qubits (typically in the triangular shape shown), and
then we need to define how checks work at the boundary. We see that
there are boundary qubits that are in two checks, and three corner
qubits that are in a single check.

\section{Outline of Proof}
The technique we use is based on the techniques used by Arora~\cite{aroraPolynomialTimeApproximation1998}
when he showed that there exist polynomial time algorithms solving
several `Euclidean' optimisation problems including the Euclidean
Travelling Salesman Problem. His proof has several novel ideas that we
adapt to the colour code decoding problem and we will point these out as we use
them.

As discussed in the introduction (Section~\ref{ss:cc}) we work in the
dual lattice and, for most of the paper, we work in the infinite
triangular lattice rather than the bounded case. If the reader would
prefer not to consider the infinite lattice then an alternative is to
think of a syndrome of an error-set that is well away from the
boundaries of a large finite colour code (and does not need to use the
boundary to be satisfied). Our key result
is the following
proposition.
\begin{proposition}\label{p:cc}
  For any $\eps>0$ there is an algorithm $\mathbf{A_\eps}$ such that
  given any syndrome $\cS$ of diameter $d$ in the infinite
  triangular lattice,
  \begin{itemize}
  \item
  $\mathbf{A_\eps}$ will output an error-set
  $\cE$ generating $\cS$, with $|\cE|\le (1+\eps)OPT$ where $OPT$ is
    the size of minimal error-set generating $\cS$.
  \item The complexity (i.e. run time) of $\mathbf{A_\eps}$ is bounded by a polynomial in $d$.
  \end{itemize}
\end{proposition}
Once we have proved Proposition~\ref{p:cc} it is relatively easy to
deduce Theorem~\ref{t:cc-bounded} which we do in Section~\ref{s:bounded_size}.

There are two key ideas for the proof of  Proposition~\ref{p:cc} (both inspired by Arora's proof).
\begin{itemize}
\item
  We set up a constrained problem -- roughly we forbid the error-set
  from containing certain errors -- and we give a (polynomial time) algorithm which
  finds the minimum weight decoding subject to these extra constraints.
\item
  We show that  any solution to the original (unconstrained)
  problem can be modified to be a solution to the constrained problem,
  without increasing its weight too much.
\end{itemize}
If we can do both of these, then we will have a polynomial time
algorithm that can find a good approximate solution -- the minimum
weight error-set for the constrained problem is exactly such a set.

\begin{figure}[t]
  \begin{subfigure}[t]{0.5\columnwidth}
  \begin{center}
    \includegraphics[width=0.95\columnwidth]{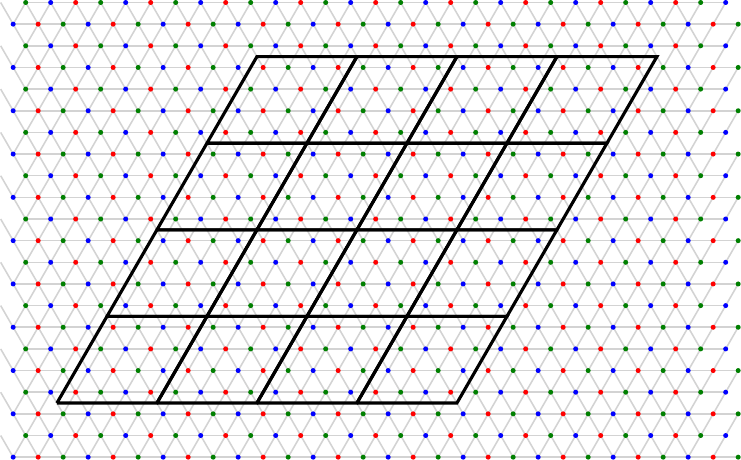}
    \caption{\label{fig:sixteen_tiles}}
  \end{center}
  \end{subfigure}
  \begin{subfigure}[t]{0.5\columnwidth}
    \begin{center}
      \includegraphics[width=0.625\columnwidth]{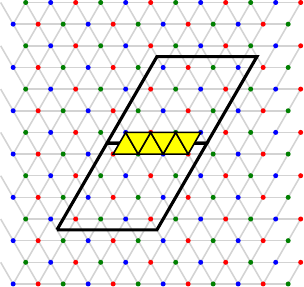}
      \caption{\label{fig:two_adjacent_tiles}}
    \end{center}
  \end{subfigure}
  \caption{\label{fig:sixteen_tiles_both}\textbf{(a)} A rhombus
    consisting of sixteen tiles. \textbf{(b)} Two adjacent tiles with
    six errors joining them: one of each colour-orientation.}
\end{figure}

To define the constrained problem we need to specify the forbidden
errors. To do this, we tile the dual lattice with $4\times 4$
rhombuses as shown in Figure~\ref{fig:sixteen_tiles}. This choice
helps simplify the proof
as it means that any two adjacent tiles have six errors joining them
-- see Figure~\ref{fig:two_adjacent_tiles} -- one of each of the six
possible colour-orientations (triangular faces can point up or down,
and there are three cyclic permutations of the colours for each of
these). When we wish to state the position of any rhombus we will
state its bottom left node, and we will use coordinates (basis) the
horizontal axis, and an axis 60 degrees up and to the right. We will
call the second of these the \emph{diagonal axis}, and we will refer
to the sides of the rhombus parallel to this axis as \emph{diagonal
sides}.

Now, given a syndrome $\cS$ of diameter $d$ pick $L=2^\ell$ such that
$2d\le L<4d$ -- so $\ell=\lfloor\log d\rfloor+1$ (all logs in this paper
are base 2).  By translation we may assume that the syndrome $\cS$ is
contained in the $L\times L$ rhombus with its bottom left corner at
the origin (for example place one point of the syndrome at the point
with coordinates $(d,d)$). We call the $L\times L$ rhombus based at
the origin the \emph{syndrome box} (so the syndrome is entirely
contained in the syndrome box).

We place a $2L\times 2L$ rhombus $R_0$ based at some point with
coordinates in $[-L]\times[-L]$, where
$[-L]=\{-1,-2,\dots,-L\}$.%
This choice means $R_0$ always contains the
syndrome box and thus the whole of the syndrome $\cS$.  We call the
base point of $R_0$ its \emph{offset}. Later we will choose the offset
randomly from inside $[-L]\times[-L]$.

We want to reduce the full problem to smaller sub-problems so that we
can recurse or use dynamic programming. Hence,  we split $R_0$ into four
rhombuses of side length $L$, and then each of these into four
sub-rhombuses of side length $L/2$ repeating until we get to single
tiles. We say that a rhombus of side length $L/2^i$ is at level $i$;
this means tiles (recall these have side length 4) are at level
$\ell-2$.

This definition might seem strange as $R_0$ is at level $-1$; however,
we will be interested in errors that cross the boundary of a rhombus,
and no error of the minimum weight error-set will cross the boundary
of $R_0$. Whilst this is intuitively true, for completeness we include
a proof that there is a minimum weight decoding of $\cS$ which is
entirely contained inside the syndrome box (so definitely inside
$R_0$) in the appendix (Lemma~\ref{l:rhombus-hexagon}) We call the
whole family of rhombuses of all sizes the rhombus-tree $\cR$.

We will need to know how the sub-problems interact with each other,
and this interaction occurs along the boundaries of the rhombuses.  We
say that an error is a \emph{boundary error of level $i$} if it
contains nodes in (at least) two distinct rhombuses at level $i$. We
will say it is a \emph{corner error of level $i$} if it contains nodes
in three distinct rhombuses at level $i$.

Observe that, if all the boundary errors of a rhombus $R$ are
specified (i.e., it is specified whether they occur or not), then we
can solve for the minimum weight error-set inside $R$ without
considering the errors or syndrome outside $R$ -- this follows
immediately because, once the boundary errors are fixed, the
checks inside $R$ are
separated from the checks outside.
Therefore, one algorithm for solving for the minimum weight
error-set inside $R_0$ is to split $R_0$ into its four child
sub-rhombuses (in the rhombus-tree), and find the optimal solution for
each of these sub-rhombuses for each of its possible boundary
configurations, and then choose the optimum over all compatible
boundaries (i.e., rhombuses that share a boundary have to agree on the
boundary). However, the number of possible boundary configurations is
exponential in the side length of $R_0$, so it is not possible to try
all of them (with a polynomial time algorithm). To get around this we
constrain which boundary errors can occur. This is very similar to
Arora's idea of `portals' in his proof and we have followed his naming
here.

We fix a parameter $M$, a power of two, which we use to control how
much we constrain the boundary errors. When we actually use $M$ it
will be of order $\log L$ but, for intuition, the reader may think of
it as a large constant for now.  Given a rhombus $R$ we define the
portal tiles of $R$ to be the four corner tiles of $R$ together with
$M-1$ tiles uniformly spaced on each of its sides (this is a little
imprecise -- we formalise the definition below). We call any level-$i$
boundary error which has all of its checks in level-$i$ portal
tiles  a \emph{level-$i$ portal error}. All other level
$i$ boundary errors are called \emph{level-$i$ forbidden errors}.

If the rhombus $R$ has less than $M$ tiles along each side, then all
the tiles along the sides are portal tiles. In this case it is
important to note that \emph{all} boundary errors are portal errors;
i.e., there are \emph{no} forbidden errors at this level. Also, note
that, since the corner tiles of any rhombus are portal tiles, all
level-$i$ corner errors are level-$i$ portal errors.

We say that an error set that does not contain any forbidden errors at
any level is \emph{portal respecting}. Note this depends on $M$ (since
the definition of portal depends on $M$), and when we wish to
emphasise this we will write \emph{$M$-portal respecting}. Also note that it
depends on the offset used to form $R_0$ since that changes the entire
tiling and rhombus tree $\cR$.

The idea of this constraint is that it limits the number of possible
boundary configurations of a rhombus to be $2^{O(M)}$ and, provided
that $M=O(\log L)$, this is polynomial in $L$.  This will allow to
solve the constrained problem recursively using dynamic programming in
polynomial time.

The following two lemmas will prove Proposition~\ref{p:cc}.

\begin{lemma}\label{l:alg}
  Given any $R_0$ and any $M$.  There is an algorithm $\mathbf{A}$
  that, for any syndrome $\cS$, finds a minimum weight
  \defnemph{$M$-portal-respecting} error-set generating $\cS$, and
  takes time at most $L^22^{72(M+1)}$.%
\end{lemma}
\begin{lemma}\label{l:mod}
  Given any syndrome $\cS$, any $M$, and any error-set $\cE$
  generating $\cS$, there  exists an offset such that $R_0$ has
  an $M$-portal respecting error-set $\cE'$ generating $\cS$ with
  $|\cE'|\le (1+ \frac{22\log L}{3M})|\cE|$.
\end{lemma}

\begin{proof}[Proof of Proposition~\ref{p:cc} assuming Lemmas~\ref{l:alg} and~\ref{l:mod}]
  
Fix a constant $C$ to be chosen later (in fact $C=22/3$ will do), and
then fix $M$ a power of two with $\frac{C}{\eps}\log L\le M <
\frac{2C}{\eps}\log L$.  Then iterate over all $L^2$ possible
offsets, using the algorithm in Lemma~\ref{l:alg} to find the minimum
weight decoding for that offset. This takes time
  \[
  O(L^2L^22^{72(M+1)})=O(L^42^{(72((2C/\eps)\log L)})=  O(L^{4+144C/\eps})
  \]
  which, for any fixed $\eps>0$, is polynomial in $L$ and, since $L\le
  4d$, polynomial in $d$.

Let $\cE_1$ be the smallest error-set that we find during this process
-- i.e., the error-set for the offset which has the minimum weight. We
claim that $\cE_1$ satisfies the bounds in Proposition~\ref{p:cc}. To
see this, let $\cE$ be the minimum weight error set for the
unconstrained problem. By Lemma~\ref{l:mod} we can find a portal
respecting (for some offset) error-set $\cE_2$ with $|\cE_2|\le (1+
\frac{22\log L}{3M})|\cE|$. Let $\cE_3$ be the minimal error set for
$\cE_2$'s offset so $|\cE_3|\le |\cE_2|$. Also, since $\cE_3$ is found by the dynamic program
when it runs for this offset we have $|\cE_1|\le|\cE_3|$.

Putting these together we get:
\[|\cE_1|\le |\cE_3|\le |\cE_2| \le \left(1+ \frac{22\log
  L}{3M}\right)|\cE| \le \left(1+ \frac{22\eps}{3C}\right)|\cE|=\left(1+ \frac{22\eps}{3C}\right)OPT
\]
which, provided we choose $C$ to be at least $22/3$, is at most
$(1+\eps)OPT$. This completes the proof of Proposition~\ref{p:cc}.
\end{proof}
We remark that this value of $C$ makes the complexity above
$O(L^{4+1056/\eps})$. Of course there are many places this bound could
be improved but, regardless, it is impractically enormous.

In the next section we make the definition of portal precise and prove
some elementary lemmas about them; then in Section~\ref{s:alg} we
prove Lemma~\ref{l:alg}, and in Section~\ref{s:mod} we prove
Lemma~\ref{l:mod}. In Section~\ref{s:bounded_size} we use these
results to prove Theorem~\ref{t:cc-bounded}. Finally in
Section~\ref{s:further-work} we discuss the result, possible
extensions, and its ramifications.

\section{Portals}
In this section we formalise exactly where the portals are, and prove
some elementary results about them. We start by defining boundary
tiles and corner tiles. For a rhombus $R$, a \emph{boundary tile} is
any tile inside $R$ that is adjacent to (at least) one side of $R$. A
\emph{corner tile} of $R$ is any tile of $R$ that is adjacent to two
sides of $R$.
\begin{defn}
  Given a rhombus $R$ of side length $4\times 2^k$ (so $2^k$ tiles on
  each side) then we define the \defnemph{portal tiles} as follows. If
  $2^k\le M$ then all boundary tiles are portal tiles. On the other
  hand, if $2^k> M$, enumerate the tiles along a side from $0$ to
  $2^k-1$, always starting with zero at the bottom left end of the
  side. The portal tiles are all the tiles that have index divisible
  by $2^k/M$ (recall $M$ is a power of two) -- equivalently they are
  the tiles whose index is congruent to zero modulo $2^k/M$, together
  with all four corner tiles of $R$.
\end{defn}
\begin{lemma}\label{l:portal-higher-levels}
  Suppose that $R$ is a rhombus and $P$ is a portal tile and that $R'$
  is a sub-rhombus containing $P$. Then $P$ is also a portal tile of
  $R'$.
\end{lemma}
\begin{proof}
  If $P$ is a corner tile of $R$ then it is a
  corner tile of $R'$ and we are done. Also if $R'$ has at most $M$
  tiles on each side, then all of its boundary tiles are portal tiles
  and again the result holds.

  Otherwise, suppose that $R$ has $2^i$ tiles on each side, and that
  $R'$ has $2^j>M$ tiles on each side. Let $v$ be the index in $R$ of
  the first tile of $R'$.  By the definition of the rhombus-tree we
  see that $v$ is a multiple of $2^j$. Let $N_R$ be the index of $P$
  inside $R$, and $N_{R'}$ the index of $P$ inside $R'$. We see that
  $N_{R'}=N_R-v$. Since $N_R$ is congruent to zero modulo $2^i/M$, it
  is also congruent to zero modulo $2^j/M$. Also  $v$ is divisible
  by $2^j$, so is congruent to zero mod $2^j/M$. Combining these, we see that
  $N_{R'}$ is congruent to zero modulo $2^j/M$: i.e., $P$ is a
  portal tile for $R'$ as claimed.
\end{proof}
\begin{lemma}\label{l:num-boundary}
  Let $R$ be a rhombus at level $i$ in the
  rhombus-tree. Then the number of portal errors for $R$ is at most $36(M+1)$
\end{lemma}
\begin{proof}There are at most $M+1$ portal tiles on each side of $R$.
  We would like to say that each portal tile has six errors linking it
  to the corresponding tile in the neighbouring rhombus. However,
  whilst this is the typical case, if two portal tiles on the side are
  neighbours there are more possibilities (e.g., the error could
  contain one defect in each of these two tiles, and one in the tile
  in the neighbouring rhombus). However it is easy to see
  (Figure~\ref{fig:two_adjacent_tiles}) that there are at most nine
  boundary errors meeting any side of a portal tile which, counting
  the four sides of the rhombus, gives the bound of $36(M+1)$.
\end{proof}
\begin{lemma}\label{l:num-interior-boundary}
  Let $R$ be a rhombus at level $i$ in the rhombus-tree, and let
  $R_1,R_2,R_3,R_4$ be its four sub-rhombus children in the
  rhombus-tree. Let $\cP$ be the set of errors that are portal errors
  for any of the sub-rhombuses but are not boundary errors for
  $R$. Then $\cP$ has size at most $36(M+1)$
\end{lemma}
\begin{proof}
  This is very similar to the previous lemma. Since $\cP$ does not
  include the boundary errors of $R$ we only need to consider the four
  `interior' sides joining the sub-rhombuses. Each of these has
  (at most) $M+1$ portal tiles. Therefore $\cP$ has size at most  $36(M+1)$.
\end{proof}
We remark that, although the bounds in the previous lemmas are
identical, they are bounding different things, and we will need both
bounds later.

The next lemma shows that certain errors are not forbidden errors --
later, when we modify a set to remove forbidden errors, we will only
use these errors and this will ensure that, when removing one
forbidden error, we do not add any new forbidden errors.
\begin{lemma}\label{l:boundary-tile}
  Suppose that $R$ is a rhombus and that $f$ is a single error wholly
  contained within tiles inside $R$ that are all adjacent to a given
  side $s$ of $R$. Then $f$ is not a forbidden error (at any level).
\end{lemma}
\begin{proof}
  By symmetry we may assume that $s$ is a horizontal side of
  $R$. Either $f$ is wholly contained within a single tile, or it
  meets two tiles which are adjacent across a diagonal side.  In the former case
  $f$ is not a boundary error for any rhombus, so is not forbidden.

  Hence, suppose $f$ meets two tiles which are `horizontally adjacent'
  and suppose that $f$ is a boundary error for some rhombus $R'$ (so
  the diagonal side of $R$ separates these two tiles).  By the definition
  of the rhombus-tree, either $R'\supseteq R$ or $R'\subset R$, and
  since $f$ is wholly contained in $R$ the first case cannot occur:
  i.e., we must have $R'\subset R$. But then we see that $R'$ has a
  horizontal side coincident with $s$ and, in particular, any tile of
  $R'$ meeting $f$ must be a corner tile, and thus a
  portal tile. Therefore $f$ is not a forbidden error
  for $R'$. Since this applies for all rhombuses $R'$ the result follows.
\end{proof}
Finally we will need to bound the total number of rhombuses in the
rhombus-tree so we include that here.
\begin{lemma}\label{l:total-rhombuses}
  The total number of rhombuses in the rhombus-tree is less than $L^2$.
\end{lemma}
\begin{proof}
  Since tiles are $4\times 4$ and $R_0$ is $2L\times 2L$ the number of
  tiles is $L^2/4$. The number of rhombuses at each level is a quarter
  of the number one level higher so the total number of rhombuses is
  then at most $L^2/4(1+1/4+1/16\dots)=L^2/3$ which is less than
  $L^2$ as claimed.
\end{proof}
\section{Algorithm}\label{s:alg}
Finding the minimum weight \emph{$M$-portal respecting} decoding is
straightforward by dynamic programming, which for the purpose of this
paper is just recursion with caching.

As we said earlier,  the restriction of the minimum weight portal
respecting decoding to any rhombus $R$ in the rhombus-tree depends
only on the syndrome inside $R$ and the portal errors for $R$.  This
is immediate because the boundary errors for $R$ disconnect the
interior of $R$ from the rest of the code.

The algorithm finds this restricted minimum weight decoding for every
possible (portal respecting) boundary configuration for every rhombus
$R$ in the rhombus-tree, by doing the following. 
Suppose that $R$'s portal errors are specified. First consider the case
when $R$ is a single tile. %
We exhaust over all possible sets of errors inside $R$, and find the
minimum weight set that, together with the boundary errors, explains
the syndrome inside $R$. There are eighteen errors wholly inside $R$
so this takes time $2^{18}$ and, since there are four sides with at
most nine errors per side there are at most $2^{36}$ boundary
configurations (this bound can easily be improved but it does not make
any difference to our final bound), so the total time for each rhombus
is $2^{54}$.

If $R$ is not a single tile, then the minimum weight decoding of the
syndrome inside $R$ can be found as follows:
  \begin{enumerate}
  \item Split $R$ into its four sub-rhombuses. These sub-rhombuses
    share some of their boundary with $R$ (which means that boundary
    is already specified), but some with each other. We call the boundaries shared by these sub-rhombuses the \emph{interior boundary}.
  \item For each possible interior boundary configuration that is
    portal-respecting, we can find the minimum weight error-set for
    each sub-rhombus and join them together to get a portal respecting
    error-set for $R$.
  \item The lowest weight example from the previous step is the
    minimum-weight error-set for $R$ given its specified boundary.
  \end{enumerate}

  Assuming that we already know (or have cached) the minimum weight
  configurations for all the sub-rhombuses for all of their possible
  boundaries, then the time to execute the steps above is bounded
  above by the number of possible interior boundary configurations. By
  Lemma~\ref{l:num-interior-boundary} there are at most $2^{36(M+1)}$
  such possibilities. We remark that since this bound is greater than
  that we gave for single tiles above, it applies in both cases.

  However, we need to calculate this minimum weight decoding for every
  possible boundary configuration for $R$. By Lemma~\ref{l:num-boundary} there are
  at most $2^{36(M+1)}$ possible boundary configurations.

  Finally, we need to do this calculation for every rhombus $R$ in the
  rhombus-tree $\cR$. Since, by Lemma~\ref{l:total-rhombuses}, there are less than $L^2$ such rhombuses, this gives
  total complexity of  $L^22^{72(M+1)}$ which completes the proof of Lemma~\ref{l:alg}.

\section{Modification}\label{s:mod}
In this section we show that we can modify an arbitrary error set to
make it portal respecting without increasing its weight too much. This
will not be true for every offset of $R_0$, but (again following
Arora's ideas for Euclidean TSP) will be true `on average' and, in
particular, will be true for at least one offset.

\begin{figure}[t]
    \begin{subfigure}[t]{0.5\columnwidth}
      \begin{center}
        \includegraphics[width=0.9\columnwidth]{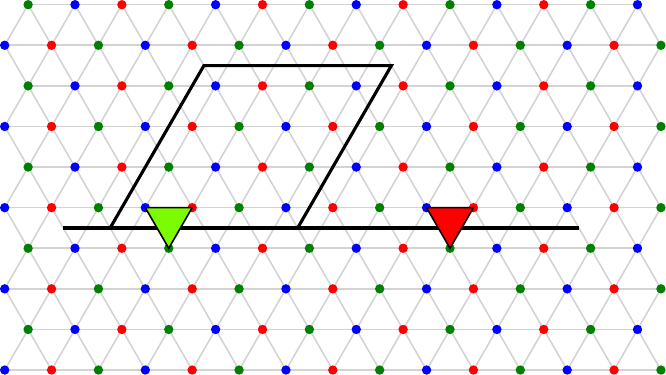}
        \caption{\label{fig:modification1}}
      \end{center}
    \end{subfigure}
    \begin{subfigure}[t]{0.5\columnwidth}
      \begin{center}
        \includegraphics[width=0.9\columnwidth]{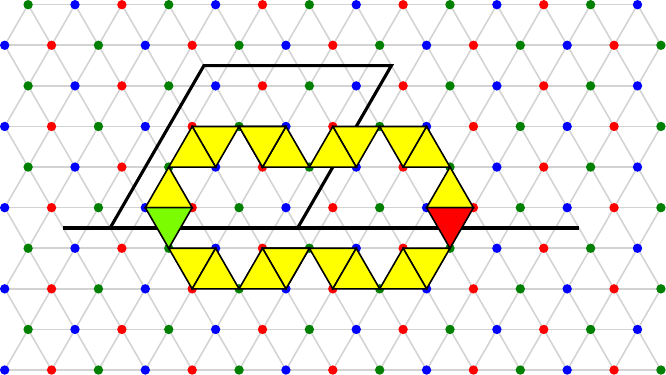}
      \end{center}
      \caption{\label{fig:modification2}}
    \end{subfigure}
  \caption{\label{fig:modification}\textbf{(a)} Given an error $f$ (the red
    triangle) look in the nearest portal tile (the rhombus shown) and
    pick the error $p$ with the same colour-orientation (the green
    triangle) in the nearest portal tile (the tile shown). \textbf{(b)} We can join
    the red error to the green error by two `paths' as shown: the set
    of all the shown errors has no syndrome, so the syndrome of the
    combined yellow and green errors (together these form $\cE_f$) is
    the same as the syndrome of the red error. }
\end{figure}

\begin{lemma}\label{l:fix-single-error}
  Let $f$ be a forbidden error (at any level), and let $i$ be the minimum level for
  which it is a boundary error (i.e., the largest size rhombus). Then there
  is a portal respecting error set $\cE_f$ with the same syndrome as
  $f$ and weight at most $\frac{11L}{3\times2^iM}=O(\frac{L}{2^iM})$.
\end{lemma}
\begin{proof}
  Let $R$ and $R'$ be the level-$i$ rhombuses meeting $f$.  Without
  loss of generality we may assume that these are adjacent across a
  horizontal side $s$. Let $p$ be the equivalent portal error to $f$
  (same orientation and colours) in the nearest portal tile to $f$
  along side $s$ -- see Figure~\ref{fig:modification1}. Suppose that
  $p$ is at distance $D$ from $f$. Since there are $M+1$ portal tiles
  on each side of $R$, and each side has length $L/2^i$, $D$ is at
  most $\frac{L}{2^iM}$. (Note that since $R$ has a forbidden error,
  not all of its boundary tiles are portals, so it must have $M+1$
  portal tiles.) We can join $p$ to $f$ using two paths of errors, one
  path in $R$ and one in $R'$ -- see
  Figure~\ref{fig:modification2}. Since each path uses four errors to
  move three across horizontally, and we also add the reflection of
  both $p$ and $f$ along their horizontal side, we see that the total
  number of errors to do this is $\frac83D+2$.  Let $\cE_f$ be the
  union of $p$ and these two paths, so $|\cE_f|=\frac83D+3\le
  \frac{11}3D\le \frac{11L}{3\cdot 2^iM}$ where the first inequality
  follows since $D\ge3$ for any forbidden error (that is the closest
  an error of the same orientation can be), and the second from the
  upper bound on $D$ above.

  It is easy to see that $\cE_f\cup \{f\}$ has no syndrome, which
  immediately implies that $\cE_f$ and $f$ have the same syndrome.
  
  Finally we need to check that $\cE_f$ is portal respecting, both in
  the horizontal direction and in the diagonal direction. Since the
  set is entirely contained in the tiles on either side of the line
  separating $R$ and $R'$, and only crosses this line at the portal
  error $p$, we see that the $\cE_f$ does not contain any forbidden
  errors crossing a horizontal side of any rhombus except possibly at
  $p$. Moreover, since we chose $i$ to be the minimum level for which
  $f$ is a boundary error, any larger rhombus containing $f$
  contains both $R$ and $R'$ and, in particular, $p$ is not a boundary
  error for such a rhombus, so cannot be a forbidden error at any
  lower level. By Lemma~\ref{l:portal-higher-levels} combined with the
  fact that $p$ is a portal error for $R$, it is a portal error for
  all higher level rhombuses which contain it. Thus, $p$ is not a
  forbidden error.

  Next we consider boundary errors in the diagonal
  direction. All the errors in the path
  joining $f$ and $p$ in $R$ lie in tiles adjacent to $s$ so, by
  Lemma~\ref{l:boundary-tile}, none of these errors can be
  forbidden. Similarly, no error in the path joining $f$ to $p$ in
  $R'$ is forbidden.  Finally, we need to check $p$ itself. Since $p$
  does not cross the diagonal boundary of any tile, it is not a
  diagonal boundary error at any level, so is definitely not a
  diagonal forbidden error.
\end{proof}

We define \emph{fixing} an error $f$ to be process whereby we replace
$f$ by $\cE_f$ in as described in Lemma~\ref{l:fix-single-error} if
$f$ is forbidden, and otherwise leave it unchanged. The \emph{cost} of
fixing an error is the bound in Lemma~\ref{l:fix-single-error} if the
error is a forbidden error which has minimum forbidden level equal to
$i$, and zero if the error is not a forbidden error.

\begin{lemma}\label{l:offset}
  Suppose that $x$ is an error. If we position the rhombus-tree
  randomly (i.e., pick the offset at random in $[-L]\times[-L]$) then the expected cost of fixing $x$ is at most
  $ \frac{22\log L}{3M}=O(\frac{\log L}{M})$.
\end{lemma}
\begin{proof}
  The chance that the error $x$ is a boundary error for a horizontal
  side of a level-$i$ rhombus is $2^i/L$ and similarly for a diagonal side. Therefore, that $x$ is a boundary error with minimum
  level $i$ is at most $2\times 2^i/L$. By Lemma~\ref{l:fix-single-error} the
  cost of fixing an error with minimum level equal to $i$ is at most
  $\frac{11L}{3\cdot2^iM}$.  Combining these two facts shows that the
  expected cost of fixing an error $x$ is at most
  \[\sum_{i=0}^{\ell-2} \frac{2\times 2^i}{L}\frac{11L}{3\cdot2^iM} = \sum_{i=0}^{\ell-2} \frac{22}{3 M}=\frac{22(\ell-1)}{3M} \le \frac{22\log L}{3M} \]
  as claimed.  
\end{proof}

\begin{proof}[Proof of Lemma~\ref{l:mod}]
  We just put things together -- we fix each forbidden error in
  turn. By Lemma~\ref{l:fix-single-error} fixing an error does not add
  any new forbidden errors so at the end of this process we have a set
  with no forbidden errors -- i.e., a portal respecting set.

  Explicitly, let $\cF$ be the set of forbidden errors in $\cE$, and set
  \[\cE'=\cE\oplus\bigoplus_{f\in \cF}\left(\cE_f\oplus f\right)=(\cE\setminus\cF)\oplus \bigoplus_{f\in \cF} \cE_f
  \]
  Using the first formulation, and the fact that the sets $\cE\oplus
  f$ have no syndrome we see that $\cE'$ has the same syndrome as
  $\cE$.  Using the second formulation, and the observing that both
  $\cE\setminus\cF$ and all the sets $\cE_f$ are portal respecting, we
  see that the set $\cE'$ is portal respecting. Finally
  \[
  |\cE'|\le |\cE|+\sum_{f\in \cF}|\cE_f|
  \]
  so
  \[
  \Exp(|\cE'|)\le |\cE|+\Exp(\sum_{f\in \cF}|\cE_f|)= |\cE|+\sum_{f\in \cF}\Exp(|\cE_f|)\le |\cE|+|\cE|\frac{22\log L}{3M}
  \]
  as claimed.
\end{proof}

\section{Bounded size codes}\label{s:bounded_size}

\begin{figure}[t]
  \begin{center}
    \includegraphics[height=0.5\columnwidth, angle=90]{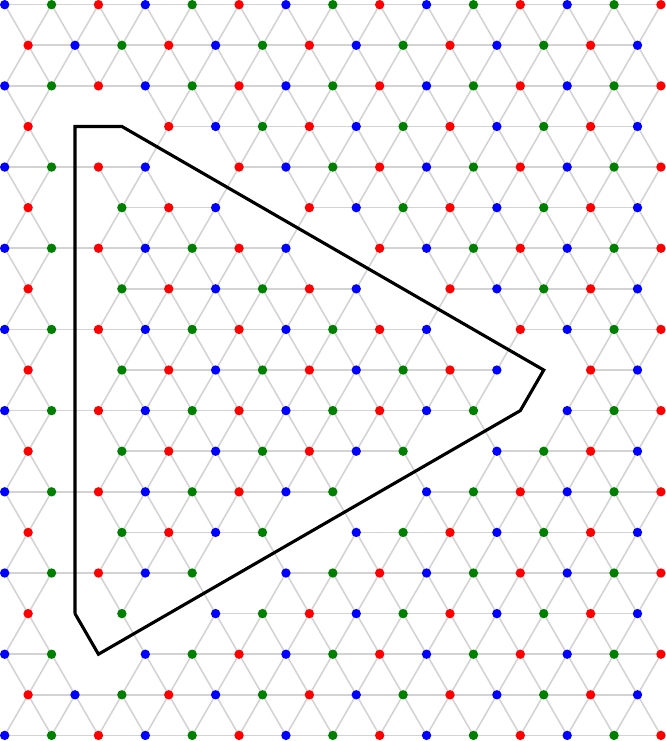}
  \end{center} 
  \caption{\label{fig:restricted_lattice}The distance 11
    reduced colour code -- the triangular region $\cT$ shows the
    embedding of the bounded colour code. We have omitted the checks
    on the boundary of the triangle as they are not enforced in the
    reduced colour code. The two key points of the reduced code
    are that it is infinite so we can apply the modifications we used
    in Section~\ref{s:mod}, and if an error set is entirely inside
    $\cT$ then its syndrome is identical in the bounded colour code
    and the reduced colour code.}
\end{figure} 

In this section we show how we can modify the techniques we have used
to prove Proposition~\ref{p:cc} to prove
Theorem~\ref{t:cc-bounded}. Essentially we just embed the bounded
problem in an infinite lattice, since then we can apply all of our
earlier techniques. This is slightly more complicated than one might
expect because the colour code with boundary is not just the
restriction of the infinite lattice to the finite triangle as there
are some faces (errors) in the infinite lattice that meet a check (node) in the triangular
region, but do not correspond to qubits in the colour code with boundary --
for example errors that meet a single node of $\cT$ (except for the
three corner errors).

We work with two different settings: the colour code with boundary,
and the reduced lattice which is an embedding of the colour code
with boundary into an infinite lattice. Whilst the colour code with
boundary can be defined for many different polygons, we will restrict
to the simplest (and standard case) of the triangular region as shown
in Figure~\ref{fig:bare-lattice}. We call this the \emph{bounded
colour code} or \emph{colour code with boundary} and we call its
domain $\cT$. We will say a face (error) is contained in $\cT$ if it
corresponds to a qubit in the bounded colour code, including the
boundary errors, and the three `corner' boundary errors.

We can think of $\cT$ as being a subset of the infinite lattice and we
can think of the boundary errors of $\cT$ as being triangles in the
infinite lattice but with some nodes not in $\cT$ and thus not being
checked. We define the \emph{reduced colour code} to be the infinite
lattice with the nodes that are in boundary errors of $\cT$ but do not
correspond to checks in $\cT$  removed. See
Figure~\ref{fig:restricted_lattice}. It is important that the checks
further out from $\cT$ are \emph{not} removed -- if we removed these
checks it would allow a single error to trigger a single detector on
the boundary of $\cT$, which is not possible in the bounded colour
code (except at a corner).

As before the proof divides into two key parts: one describing the
polynomial time algorithm and one showing that that algorithm does
produce a set which is a $1+\eps$ approximation to the minimum weight
decoding.  In more detail the steps are:
\begin{enumerate}
\item The algorithm:
  \begin{enumerate}
\item\label{step-find-min-weight} Find the minimum-weight portal
  respecting error-set $\cE_2$ generating $\cS$ in the reduced
  colour code. This can be done in polynomial time using the same
  dynamic programming algorithm as described in Section~\ref{s:alg}
  except that we do not impose any constraint
  on the parity of the number of errors at any of the removed checks.
\item \label{step:post-process} `Post-process' $\cE_2$ forming a new
  set $\cE_3$ such that
  \begin{itemize}
  \item
    $|\cE_3|\le |\cE_2|$, 
  \item 
    all the errors in $\cE_3$ lie in $\cT$,
  \item $\cE_3$ also generates $\cS$ in the reduced code,
  \end{itemize}
  Then $\cE_3$ generates $\cS$ in the bounded code.
  \end{enumerate}
\item \label{step-bounded-to-restricted}To prove that the set the
  algorithm finds is a $1+\eps$ approximation, it suffices to show
  that, given a syndrome $\cS$ in the bounded colour code, there
  exists a portal respecting error-set in the \emph{reduced colour
  code} that has weight at most $(1+\eps)OPT$ where $OPT$ is the
  minimum weight decoding in the bounded colour code (since then
  $\cE_3$ has weight at most this).  This can be done with the
  following sub-steps.
  \begin{enumerate}
  \item Let $\cE$ be the minimum weight error-set for the syndrome in
    the bounded colour code.  Then $\cE$ is a solution to the syndrome
    $\cS$ in the reduced colour code (in fact it turns out to be a minimum
    weight solution but we do not need that fact).
  \item Form a modified set $\cE_1$ that is a portal-respecting error-set
  for $\cS$ in the reduced lattice. This is exactly from the main
  body of the paper. The missing checks do not affect this step. We
  know $|\cE_1|\le (1+\eps)|\cE|$.
  \end{enumerate}

\end{enumerate}
The set $\cE_3$ from Step~\ref{step:post-process} is our approximate
solution -- indeed we have
\[|\cE_3|\le |\cE_2|\le|\cE_1|\le(1+\eps)|\cE|=(1+\eps)OPT\]
where the first and third inequality follow from the properties listed
above, and the second inequality follows since $\cE_2$ is the minimal
portal respecting error-set so has size at most that of $\cE_1$ (since
$\cE_1$ is portal respecting). We see that it is important that the
post-processing step only take polynomial time (so the whole algorithm
remains polynomial. Also note that $\cE_3$ does \emph{not} need to be
portal respecting.

Step~\ref{step-find-min-weight} and  Step~\ref{step-bounded-to-restricted} are from the first part
of this paper.  The one part that remains is
Step~\ref{step:post-process}: we need to show that we can do the
post-processing. We start by showing that we can force the errors to
lie in a half-plane.

\begin{figure}[t]
      \begin{subfigure}[t]{0.5\columnwidth}
      \begin{center}
        \includegraphics[height=0.9\columnwidth, angle=-90]{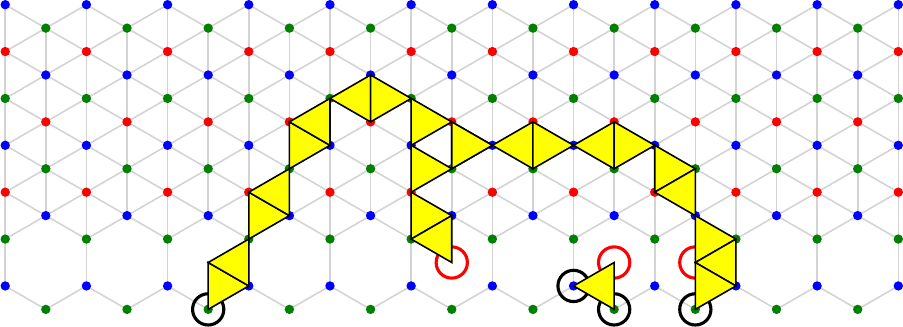}
        \caption{\label{fig:projection-a}}
      \end{center}
      \end{subfigure}
    \begin{subfigure}[t]{0.5\columnwidth}
      \begin{center}
        \includegraphics[height=0.9\columnwidth, angle=-90]{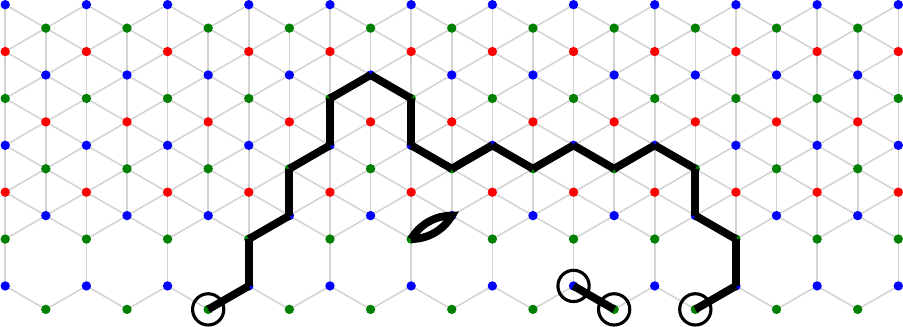}
        \caption{\label{fig:projection-b}}
      \end{center}
    \end{subfigure}

    \vspace{1cm}
    \begin{subfigure}[t]{0.5\columnwidth}
      \begin{center}
        \includegraphics[height=0.9\columnwidth, angle=-90]{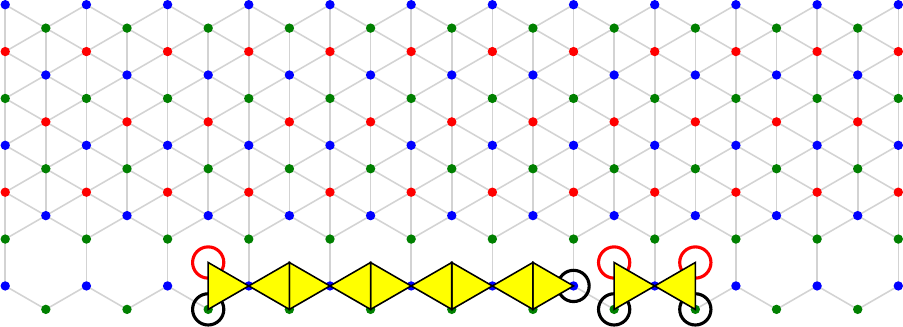}
        \caption{\label{fig:projection-c}}
      \end{center}
    \end{subfigure}
    \begin{subfigure}[t]{0.5\columnwidth}
      \begin{center}
        \includegraphics[height=0.9\columnwidth, angle=-90]{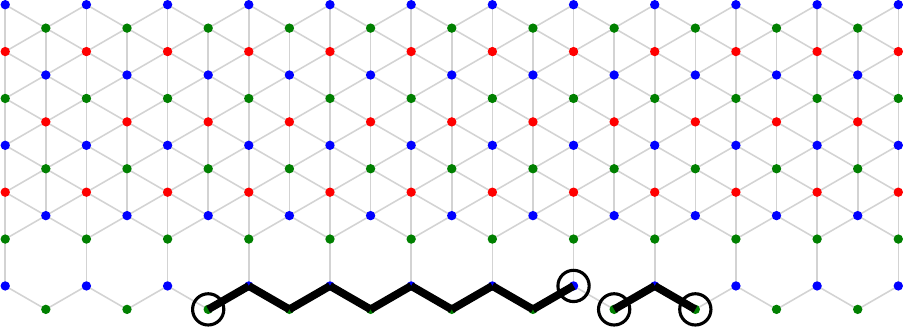}
        \caption{\label{fig:projection-d}}
      \end{center}
    \end{subfigure}
  \caption{\label{fig:projection}The four steps in the post-processing described in
    Lemma~\ref{l:project-half-plane}; working clockwise from the top
    left: \textbf{(a)} we start with the original error-set.  Top right \textbf{(b)}: we have the
    error-set viewed as edges; note one edge (shown as two arcs) is
    actually a double edge so cancels with itself.  Then bottom right \textbf{(d)} we
    have the projected edge error-set, and then bottom left \textbf{(c)} the actual
    post-processed error-set. We see that the syndrome of the two sets
    is the same except for the third row (nodes circled in red). }
\end{figure}
 
\begin{lemma}\label{l:project-half-plane}
  Suppose that $\cE$ is a set of errors in a half plane oriented as in
  Figure~\ref{fig:projection-a} with syndrome entirely contained in the
  first three rows of vertices (one row of each colour). Then there is
  an error set $\cE'$ which is entirely contained in these three rows
  such that $|\cE'|\le |\cE|$ and that the syndrome of $\cE'$ is equal
  to the syndrome of $\cE$ except (possibly) on the third row.
\end{lemma}
 The whole of the process constructing $\cE'$ is shown in
 Figure~\ref{fig:projection} and that may be easier to follow than the
 formal proof below. Note in the figure we have omitted the checks on
 the third row as we do not require the syndrome to agree on that row.
\begin{proof}
  Consider just the vertices in the half plane which are not the
  colour of the third row. Every error in $\cE$ contains exactly two
  of these vertices (and one of the vertices of the colour we are
  ignoring) so this naturally forms a graph $G$ on these vertices
  (some edges could be double edges if inside two errors) -- see
  Figure~\ref{fig:projection-b}. The vertices of odd degree in this
  graph are exactly the defects that do not occur in the omitted
  colour, so all lie in the first two rows of the
  graph. We form a new graph -- in fact, a multi-graph with loops --
  by projecting each vertex onto the corresponding vertex in the first
  two rows.  Since none of the projected vertices are in the first two
  rows, they all have have even degree. Therefore the parity of the
  degree of every vertex in the first two rows is unchanged by the
  projection.  Since loops are degree two we can remove them without
  changing the parity of any vertex. Similarly we can replace any
  multi-edge by zero or one edges preserving the parity of the number
  of edges -- see Figure~\ref{fig:projection-d}. We now have a genuine
  graph $G'$ (no loops or multi-edges) with exactly the same vertices
  having odd degree as in $G$
  and the number of edges at most the number of edges in $G$ which is
  as most $|\cE|$. Let $\cE'$ be the set formed by replacing each edge
  of $G'$ by the unique face (error) in the half plane containing it
  then we have the desired set $\cE'$ -- see Figure~\ref{fig:projection-c}. Obviously, the
  syndrome of $\cE'$ is contained in the first three rows. In the
  non-omitted colour the defects are exactly the odd degree vertices
  in $G'$ which are the odd degree vertices in $G$, and in turn are
  the defects in the non-omitted colour in the sydnrome of $\cE$. By
  construction and hypothesis the syndrome of both $\cE$ and $\cE'$ in
  the omitted colour all lie in the third row of the lattice. Hence,
  we see that the set $\cE'$ satisfies the properties claimed.

\end{proof}
Now it is easy to prove that we can do the post-processing step.
\begin{lemma}
  Suppose that $\cS$ is a syndrome entirely contained in $\cT$, and
  that $\cE$ is an error-set generating $\cS$ in the reduced
  lattice. Then there is an error-set $\cE'$ with all of its errors in
  $\cT$, that generates $\cS$, and has $|\cE'|\le |\cE|$. Moreover,
  given $|\cE|$, the set $\cE'$ can be found in polynomial time.
\end{lemma}
\begin{proof} 
  We think of $\cT$ as being the intersection of three half-planes. We
  modify the error-set $\cE$ to avoid each half plane in turn as
  described in the previous lemma. Call this set $\cE_1$. We would
  like to set $\cE'=\cE_1$, but we need to be a little careful since the
  modification in one direction may break what we fixed on the earlier
  sides -- in other words $\cE_1$ may contain errors outside
  $\cT$. The errors outside of $\cT$ cannot be arbitrarily placed --
  they must all be on one of lines extending a side of $\cT$. In
  particular, all such errors are only connected to the errors in
  $\cT$ through a corner of $\cT$. Since these corners are `omitted'
  checks in the reduced lattice, setting $\cE'$ to be all the
  errors in $\cE_1$ that are contained in $\cT$ gives the required
  set.
\end{proof}

\section{Conclusion and Further Work}\label{s:further-work}
In this paper we have shown that we can approximately decode the
$6.6.6$ code capacity colour code with $X$-errors in (enormous) polynomial
time. The key point of our work is that decoding approximately is \emph{not}
NP-hard, and this opens the door to `practical' polynomial time
algorithms approximately decoding the colour code. 

This naturally raises the question of whether the ideas in our proof
can be modified to give a practical polynomial. We chose to present
the simplest algorithm and proof, without aiming to optimise the
bounds. Indeed, the bounds we have given can be very substantially
improved, but would still be impractically large. In fact, by limiting
the number of portals that are used in any rhombus (again following
Arora), the degree of polynomial can be reduced to quasi-linear (order
$d(\log d)^{O(1)}$) -- however, the leading coefficient would be
enormous and the algorithm would still be impractical, we defer this
to a later paper.

\begin{question}
  Is there a \defnemph{practical} polynomial time algorithm approximating the
  minimum weight decoding to within a factor $1+\eps$?
\end{question}
Note, the fact that exact decoding is NP-hard shows that the runtime
must grow super-polynomially in $1/\eps$, so this is unlikely to be
practical for all $\eps$. But it is entirely possible that there is a
sensible algorithm approximating to within $1\%$ say.

As we remarked in the introduction, the restriction decoder of
Delfosse and Kubica~\cite{kubicaEfficientColorCode2023} does guarantee
(under code capacity noise) to give a decoding which is at most three
times the minimum weight.  So perhaps the following question, whilst
imprecise, is better.

\begin{question}
  For what $\eps$ is there a practical polynomial time algorithm approximating the
  minimum weight decoding to within a factor $1+\eps$?
\end{question}
One natural candidate is the M\"obius decoder of Sahay and
Brown~\cite{sahayDecoderTriangularColor2022}. As far as we are aware
it is plausible (but not proven) that that approximates the minimum
weight decoding to within a factor of $7/6$.

The second question is to what extent these results can be modified or extended to
apply to other codes and noise models.  Some of these follow with
minor modifications to this work -- for example, many two dimensional
problems could be dealt with using similar techniques. In particular, it
is simple to modify our proof to the case to depolarising noise;
essentially we just allow each face to be in one of four states --
either no error or one of $X,Y,Z$ error -- and then solve in the same
way.

However, phenomenological noise or circuit level noise would naturally
be three dimensional and the techniques we have used, as presented in
this paper, do not apply there. The
  first issue is that, in order for Lemma~\ref{l:mod} to
  work we need more portals -- roughly $\log L$ in each direction. As
  it stands that would make the bound of order $2^{(\log L)^2}$ which
  is super-polynomial. This can be avoided with a little care by
  limiting the number of portals that are actually used (as Arora did
  in his proof for Euclidean TSP in higher dimensions). The second
  issue is that we need to make sure that the modifications we do in
  one direction to make the set portal respecting do not break what we
  have done in other directions (the higher dimensional analogue of
  Lemma~\ref{l:boundary-tile}). Whilst this can be done it is rather more intricate.

  We have been able to extend our results to apply to a large class of
  codes that embed in Euclidean space, including the colour code with
  circuit level noise as discussed above, and a paper detailing these
  results is in preparation.

A very general class of codes of this type is that of geometrically
local codes. Roughly, these are codes that embed in a lattice in
$\R^D$ for some $D$ with the property that all checks have diameter
bounded by some constant independent of the code distance. These are
the codes for which the famous BPT~\cite{Bravyi2010} bound on the
encoding rate applies.

\begin{question}
  Is it the case that for every geometrically local code (in any
  dimensional Euclidean space) and every $\eps>0$ there is a
  polynomial time algorithm $\mathbf{A_\eps}$ approximating the
  minimum weight decoding to within a factor of $1+\eps$?
\end{question}

\section{Acknowledgements}

We would like to thank Ben Barber for discussions on the underlying
combinatorics, Mark Turner for discussions on the colour code, and
Joan Camps for further comments on the manuscript.  Finally, we thank
Steve Brierley and our colleagues at Riverlane for creating a
stimulating environment for research.

\bibliographystyle{plain} 
 
\bibliography{articles_tidied.bib} 

\section*{Appendix}
\begin{lemma}\label{l:rhombus-hexagon}
  Suppose that $\cS$ is a syndrome of diameter $d$ lying inside the $L\times L$ syndrome box and including the point $(d,d)$. Then there is a minimum
  weight error-set $\cE$ generating $\cS$ where all the errors of
  $\cE$ lie inside the syndrome box.
\end{lemma}
\begin{proof}
  Let $H$ be the regular hexagon of `radius' $d$ centred at the point
  $(d,d)$ -- see Figure~\ref{fig:rhombus-hexagon}. We see that $H$ is
  contained the syndrome box and that, by the diameter constraint, the
  entire syndrome is inside $H$. Thus it suffices to show that there
  is a minimum weight error-set contained entirely inside $H$.
  Suppose we have any error-set $\cE'$ in the infinite lattice
  generating $\cS$. We can reflect the errors along the (infinite)
  line extending any side of $H$ and, since there are no defects in
  the reflected region this does not change the syndrome. This
  reflection does not increase the number of errors in the error-set
  (it can decrease the number if an error and its reflection are both
  in $\cE'$ and cancel).

  To complete the proof we just need to show that we can do an
  appropriate sequence of reflections to make all the errors lie in
  $H$. Since the reflection process never increases the number of
  errors outside $H$ we just need to show that we can guarantee to
  reduce the number of errors by some collection of reflections. To do
  this just consider an error $e$ not in $H$ and the infinite
  tessellation of the plane by translates of $H$. The error $e$ lies
  in one of these tiles $H'$ say. It suffices to show that we can
  apply a sequence of reflections along the edges of $H$ which maps
  $H'$ to $H$. One way of seeing this is to think of $H$ as a
  hexagonal coin and consider flipping it repeatedly over any of its
  edges until it gets to $H'$. This corresponds to a sequence of
  reflections mapping the error $e$ into~$H$.
  \end{proof}
\begin{figure}[t]
  \begin{center}
  \includegraphics[width=0.8\columnwidth]{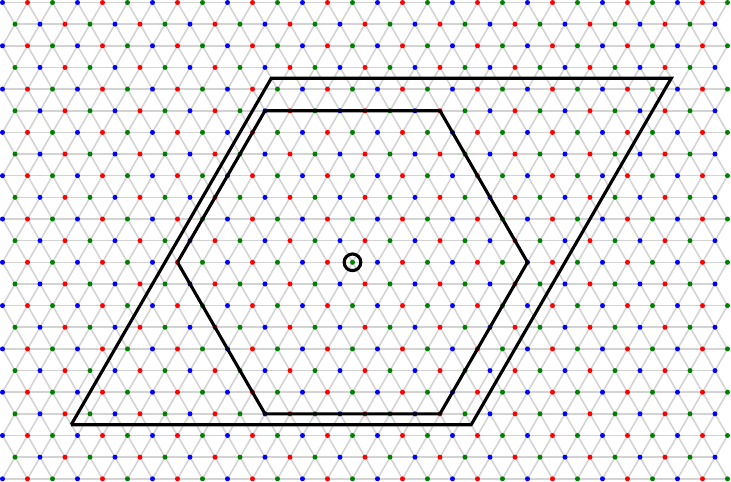}
    \end{center}
  \caption{\label{fig:rhombus-hexagon}The syndrome box (the rhombus) with the hexagon $H$ inside
    it. The hexagon is oriented so that the sides are along lattice
    lines and, in particular, we see that the lattice has reflection
    symmetry about any of the sides. }
\end{figure}

\end{document}